\documentclass[twocolumn,secnumarabic,amsmath,amsfonts,amssymb, nobibnotes, aps, prb, superscriptaddress,floatfix,footinbib,reprint]{revtex4-2}

\usepackage{amsthm}
\usepackage{color}
\usepackage{graphicx}
\usepackage{hyperref}
\usepackage{newtxtext}
\usepackage{tikz}
\usepackage{ulem}
\usepackage{xcolor}
\usetikzlibrary{shapes.geometric,arrows}
\usetikzlibrary{positioning}
\usetikzlibrary{calc}
\hypersetup{colorlinks,breaklinks}
\hypersetup{
	colorlinks=true,
	linkcolor=blue,
	filecolor=blue,      
	urlcolor=blue,
	citecolor=blue,
}

\tikzstyle{startstop} = [rectangle, rounded corners, minimum width = 0.2cm, minimum height=0.2cm,text centered, draw = black, fill = red!40]
\tikzstyle{io} = [trapezium, rounded corners, trapezium left angle=70, trapezium right angle=110, minimum width=1cm, minimum height=0.5cm, text centered, draw=black, fill = blue!40]
\tikzstyle{process} = [rectangle, rounded corners, minimum width=0.5cm, minimum height=0.4cm, text centered, draw=black, fill = yellow!50]
\tikzstyle{decision} = [diamond, rounded corners, aspect=3.7,text centered, draw=black, align = center, minimum width = 0.5cm, minimum height = 0.4cm,fill = green!30]
\tikzstyle{arrow} = [->,>=stealth,rounded corners]

\newcommand{\ab}{\textit{ab initio}}
\newcommand{\eps}{\varepsilon}
\newcommand{\etal}{\textit{et al.}}
\newcommand{\F}{\mathrm{F}}
\newcommand{\N}{\mathbb{N}}
\newcommand{\R}{\mathbb{R}}

\newcommand{\T}{\top}
\newcommand{\trace}{\mathrm{Tr}}

\newcommand{\lrbrace}[1]{\left\{#1\right\}}
\newcommand{\lrbracket}[1]{\left(#1\right)}

\newcommand{\white}[1]{\textbf{\color{white}{#1}}}
\newcommand{\abs}[1]{\left|#1\right|}
\newcommand{\inner}[1]{\left\langle#1\right\rangle}
\newcommand{\mb}[1]{\mathbf{#1}}

\newcommand{\snorm}[1]{\Vert#1\Vert}

\DeclareMathOperator*{\argmin}{arg\,min}

\newtheorem{thm}{Theorem}
\newtheorem{lem}{Lemma}
\begin{document}
	
	\title{A force-based gradient descent method for \textit{ab initio} atomic structure relaxation}%
	
	\author{Yukuan Hu}%
	\affiliation{State Key Laboratory of Scientific and Engineering Computing, Academy of Mathematics and Systems Science, Chinese Academy of Sciences, Beijing 100190, China}%
	\affiliation{University of Chinese Academy of Sciences, Beijing 100049, China}%
	\author{Xingyu Gao}%
	\email[Corresponding author: ]{gao\_xingyu@iapcm.ac.cn}
	\affiliation{Laboratory of Computational Physics, Institute of Applied Physics and Computational Mathematics, Beijing 100088, China}%
    \author{Yafan Zhao}%
	\affiliation{CAEP Software Center for High Performance Numerical Simulation, Beijing 100088, China}%
	\author{Xin Liu}%
	\email[Corresponding author: ]{liuxin@lsec.cc.ac.cn}
	\affiliation{State Key Laboratory of Scientific and Engineering Computing, Academy of Mathematics and Systems Science, Chinese Academy of Sciences, Beijing 100190, China}%
	\affiliation{University of Chinese Academy of Sciences, Beijing 100049, China}%
	\author{Haifeng Song}%
	\email[Corresponding author: ]{song\_haifeng@iapcm.ac.cn}
	\affiliation{Laboratory of Computational Physics, Institute of Applied Physics and Computational Mathematics, Beijing 100088, China}%
	
	\begin{abstract}
		Force-based algorithms for \ab~atomic structure relaxation, such as conjugate gradient methods, usually get stuck in the line minimization processes along search directions, where expensive \ab~calculations are triggered frequently to test trial positions before locating the next iterate. We present a force-based gradient descent method, WANBB, that circumvents the deficiency. At each iteration, WANBB enters the line minimization process with a trial stepsize capturing the local curvature of the energy surface. The exit is controlled by an unrestrictive criterion that tends to accept early trials. These two ingredients streamline the line minimization process in WANBB. The numerical simulations on nearly 80 systems with good universality demonstrate the considerable compression of WANBB on the cost for the unaccepted trials compared with conjugate gradient methods. We also observe across the board significant and universal speedups as well as the superior robustness of WANBB over several widely used methods. The latter point is theoretically established. The implementation of WANBB is pretty simple, in that no a priori physical knowledge is required and only two parameters are present without tuning.
	\end{abstract}
	
	\maketitle
	
    \par Atomic structure relaxation determines the ground-state atomic configuration by searching for the local minimum in the energy landscape. It makes the foundation and potential performance bottleneck for the search of global structure minimum \cite{oganov2006crystal,wang2010crystal,chen2017sgo} and high-throughput calculations in material design \cite{goedecker2005global,vilhelmsen2012systematic,hu2013pressure,li2013global,zhang2013genetic}. Conjugate gradient methods (CG) \cite{shewchuk1994introduction}, the direct inversion in the iterative subspace (DIIS) \cite{pulay1980diis}, and limited-memory Broyden-Fletcher-Goldfarb-Shanno quasi-Newton methods (LBFGS) \cite{liu1989limited} are three widely used methods to relax the atomic structure following the density functional theory (DFT) calculation. While DIIS and LBFGS could be slightly faster in some cases, CG is more stable in the situations where the initial atomic configuration is far from equilibria. For large systems, hundreds of steps may be required to converge to the typical atomic force tolerance of 0.01 eV/\AA. So it will be tremendously helpful even if the \ab~atomic relaxation could be sped up to some degree.
	
	\par One reason for the slow convergence of the local minimization goes to the narrowly curved energy valley, which prevents the efficient execution of the conventional CG. Thus, various preconditioners have been designed for improving search directions via approximating the Hessian matrices. Zhao \etal~\cite{zhao2006motif} notice that the Hessian matrix of large quantum dots can be approximated by that of motifs comprising merely a few atoms. Some other works construct general force fields based on a priori physical knowledge. Surrogate potentials are used to either directly approximate the Hessian matrix \cite{goedecker2001linear,fernandez2003model,packwood2016universal}, possibly with on-the-flight fittings \cite{chen2014approximate,chen2017curved}, or generate the ``preconditioned force'' with the implicit inversion of the Hessian matrix \cite{liu2018force}. 

    \par Another reason is the inefficient execution of the line minimization (LM) to which not much attention has been paid in the previous study. Actually, gradient descent method and CG are two typical LM-based methods \cite{shewchuk1994introduction,nocedal2006numerical,press2007numerical}. A flowchart of the general LM-based method for the \ab~atomic relaxation is depicted in Fig. \ref{fig:LS flowchart}. 
    \begin{figure}[htbp]
    	\centering
    	\begin{tikzpicture}[node distance=1.2cm]
    		\draw [dashed,rounded corners,fill = gray!30] (-3cm,-4.2cm) rectangle (3.9cm,-9cm);
    		\node [startstop,minimum width = 1.5cm] (start) {Start};
    		\node [io, below of = start, align = center,yshift=0.2cm] (input) {Input initial atomic positions $R_0$,\\force tolerance $\eps>0$; let $k:=0$};
    		\node [decision, below of = input, yshift=-0.5cm] (tolerance) {Force tolerance satisfied?};
    		\node [startstop, right of = tolerance, xshift=2.5cm] (stop) {Stop};
    		\node [process, below of = tolerance] (updatedir) {Update search direction $D_k$};
    		\node [process, below of = updatedir,align = center] (compstep) {Compute \textit{initial trial stepsize} $\alpha_k^{\text{trial}}$;\\let scaling factor $r_k:=1$};
    		\node [process, below of = compstep, align = center,yshift=-0.1cm] (updatefactorpos) {Use an LM algorithm to update $r_k$;\\compute new trial $R_k^{\text{trial}}:=R_k+r_k\alpha_k^{\text{trial}}D_k$};
    		\node [decision, below of = updatefactorpos,yshift=-.6cm] (ls) {\textit{LM criterion} accepts trial?};
    		\node [process, below of = ls,yshift=-0.3cm] (updatepos) {Update atomic positions $R_{k+1}$; let $k:=k+1$};
    		\draw [arrow] (start) -- (input);
    		\draw [arrow] (input) -- node [process,right,rounded corners = 0pt,fill=blue!80,xshift=0.1cm] {\white{KS}} (tolerance);
    		\draw [arrow] (tolerance) -- node [above] {Yes} (stop);
    		\draw [arrow] (tolerance) -- node [right] {No} (updatedir);
    		\draw [arrow] (updatedir) -- (compstep);
    		\draw [arrow] (compstep) -- (updatefactorpos);
    		\draw [arrow] (updatefactorpos) -- (ls);
    		\draw [arrow] (ls) -- node [right] {Yes} (updatepos);
    		\draw [arrow] (updatepos.west) -| ++(-0.5cm,0) - ++ (0,6cm) |- (tolerance.west); 
    		\draw [arrow] (ls.east) -| ++ (0.5cm,0) - ++ (0,1.3cm) |- (updatefactorpos.east);
    		\node [right of = ls, xshift = 2.4cm, yshift = 0.15cm] {No};
    		\node [process, below of = updatefactorpos, xshift = 0.4cm, yshift = .45cm, rounded corners = 0pt, fill = blue!80] {\white{KS}};
    		\node at (3cm,-4.45cm) {\textit{LM process}};
    		\node at (-2.4cm,-3.3cm) (existing) {\textbf{Existing works}};
    		\node at (3.0cm,-3.5cm) (this) {\textbf{This work}};
    		\draw [arrow] (existing.south) -- (updatedir.west);
    		\draw [arrow] (this.south) -- (2.5cm,-4.2cm);
    	\end{tikzpicture}
    	\caption{A flowchart of the general LM-based method, where $N\in\N$ refers to the number of atoms, $R_k\in\R^{3\times N}$ stands for the Cartesian atomic coordinates, $F_k\in\R^{3\times N}$ denotes the atomic forces applied at $R_k$, and $D_k\in\R^{3\times N}$ is the search direction. It should be noted that the white bold text ``\textbf{KS}'' refers to solving the corresponding Kohn-Sham equations to obtain new energies and forces whenever the atomic positions are updated. The LM process is marked out by the shaded box.}
    	\label{fig:LS flowchart}
    \end{figure}
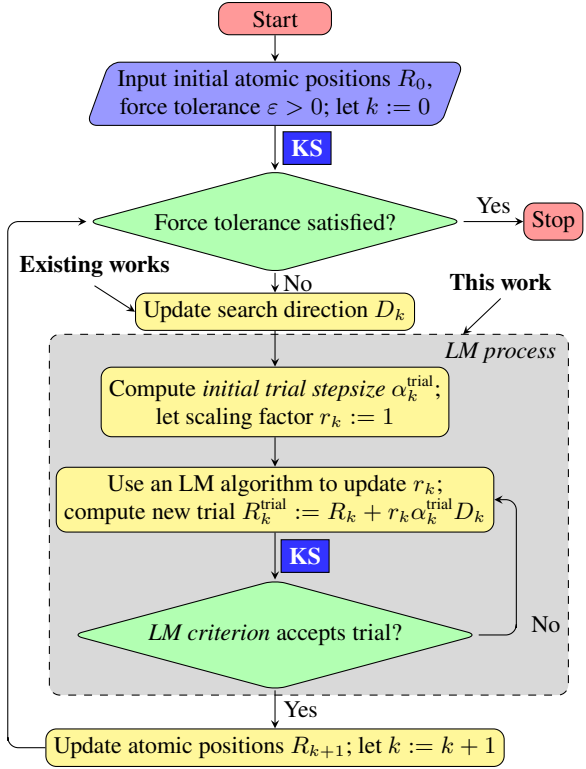 
    The \ab~LM process, marked out by the shaded box in Fig. \ref{fig:LS flowchart}, is carried out along search directions such as steepest descent or CG. During this process, one begins with an initial trial stepsize and then repeatedly invokes an LM algorithm to generate new trials until certain LM criterion is met. Involving moving atoms, each trial step is followed by \ab~calculation, e.g., solving the Kohn-Sham (KS) equation. But only the last accepted trial step triggers the update of search direction. Taking the conventional CG for instance, we test nearly 80 systems and find that each update of search direction is followed by 2.4 calls of LM algorithms on average. In other words, nearly 60\% of the total computational overhead is paid for the unaccepted trial steps in the LM process. Moreover, preconditioning with the (fitted) classical force fields does not necessarily waive the cost in the LM process \cite{chen2017curved}. As a result, the computational cost on the LM process is far more than that on simply determining a search direction if the whole procedure gets frequently trapped in the inner loop.
	
	\par In this work, we present an LM-based gradient descent method, called WANBB, which avoids trapping in the LM process during atomic structure relaxation. We achieve this by (i) calculating an initial trial stepsize that captures the local curvature of potential energy surface (PES) and (ii) devising an unrestrictive LM criterion that tends to accept early trials. Nearly 80 systems have been tested, including organic molecules, metallic systems, semiconductors, surface-molecule adsorption systems, ABX$_3$ perovskites, etc. The set of benchmark tests favors the universality of everyday \ab~atomic relaxation. The unaccepted trial steps in the LM process account for only about 1\% of the total solved KS equations on average in WANBB, compared with nearly 60\% in CG. The considerable compression of the LM process leads to a prominent saving on running time. The average speedup factors of WANBB over CG, DIIS, and LBFGS are about 1.5, 1.2, and 1.2, respectively. The robustness of the atomic relaxation method, i.e., convergence to equilibria regardless of initial configurations, is also investigated. While CG, DIIS, and LBFGS fail on some systems, WANBB manages to converge across the benchmark. This robustness is theoretically established and helpful for the cases with local lattice distortion or reconstruction. Last but not least, WANBB is pretty simple. Free of a priori physical knowledge, it can work with only two parameters present but no tuning required.
	
	\par In what follows, we are ready to deliver the algorithmic development. To facilitate narration, we collect some notations beforehand. We denote by $E(R)\in\R$ and $F(R)\in\R^{3\times N}$, respectively, the potential energy and atomic forces evaluated at $R$. When describing algorithms, we use subscripts for abbreviation; e.g., both $E_k$ and $F_k$ are evaluated at $R_k$. We denote by $\mb{I}$ the identity matrix in $\R^{3\times N}$. The operator $\inner{\cdot,\cdot}$ represents the Frobenius inner product of two $3$-by-$N$ matrices
	$$\inner{M,N}:=\trace(M^\T N),\quad\forall~M,N\in\R^{3\times N},$$
	whereas $\snorm{\cdot}_\F$ yields the Frobenius norm as $\snorm{M}_\F:=\sqrt{\inner{M,M}}$ for any $M\in\R^{3\times N}$. We further denote by $\snorm{M}_{2,\infty}$ the $\ell_{2,\infty}$-norm of matrix $M$, namely, the maximum among the $\ell_2$-norms of columns in $M$. In particular, $\snorm{F}_{2,\infty}$ gives the maximum $\ell_2$-norm among the atomic forces. Throughout iterations, we take the atomic forces $F_k$ as the search direction $D_k$. 
    
    \par The whole procedure shown in Fig. \ref{fig:LS flowchart} can fall into the trap of the LM process owing to two factors. One is the poorly chosen initial trial stepsize, which degrades the starting states of the LM process. The other is the restrictive LM criterion, which acts as a stringent stopping rule for the LM algorithm. The ideal situation is that the procedure always passes the test of the LM criterion with just the original initial stepsize. Intuitively, we require both good initial stepsizes and a loose but safe criterion to achieve this target. Here, ``safe'' means that the convergence property is maintained with the LM criterion.
	
	\par At first sight, a good initial stepsize should grasp, as much as possible, the curvature of PES around the current configuration. As early as 1988, Barzilai and Borwein \cite{barzilai1988two} propose the following renowned Barzilai-Borwein (BB) stepsizes: for any $k\ge1$,
	$$\alpha^{\text{BB1}}_k:=\frac{\inner{S_{k-1},S_{k-1}}}{\inner{S_{k-1},Y_{k-1}}},\quad\alpha^{\text{BB2}}_k:=\frac{\inner{S_{k-1},Y_{k-1}}}{\inner{Y_{k-1},Y_{k-1}}},$$
	where $S_{k-1}:=R_k-R_{k-1}$ and $Y_{k-1}:=F_{k-1}-F_k$ represent the difference of successive atomic positions and forces, respectively. Remarkably, these two stepsizes satisfy 
	\begin{align*}
		\alpha^{\text{BB1}}_k=&\argmin_\alpha\snorm{\alpha^{-1}S_{k-1}-Y_{k-1}}_\F,\\
		\alpha^{\text{BB2}}_k=&\argmin_\alpha\snorm{\alpha Y_{k-1}-S_{k-1}}_\F.
	\end{align*}
    In other words, $\lrbracket{\alpha^{\text{BB1}}_k}^{-1}\mb{I}$ and $\lrbracket{\alpha^{\text{BB2}}_k}^{-1}\mb{I}$ approximate the Hessian matrix around $R_k$ along the last search direction. They capture the local curvature of PES with merely two $3$-by-$N$ matrices. The alternating BB (ABB) stepsize \cite{dai2005projected} is presented as follows
	\begin{equation}
		\alpha^{\text{ABB}}_k:=\left\{\begin{array}{ll}
			\alpha^{\text{BB1}}_k, & \text{if}~\mathrm{mod}(k,2)=1,\\
			\alpha^{\text{BB2}}_k, & \text{if}~\mathrm{mod}(k,2)=0,
		\end{array}\right.
		\label{eqn:abb}
	\end{equation}
	which is shown to overwhelm BB1 and BB2 in many contexts (e.g., \cite{dai2005projected,gao2019parallelizable}). It is not difficult to prove that the implementation of (A)BB stepsizes requires the same computational and storage complexity as computing the conjugate parameter in CG.
	
	\par Despite the simplicity, the (A)BB stepsizes turn out to be significant benefits in solving various optimization problems. The most related application goes to solving KS equations as orthogonality constrained optimization problems, where the (A)BB stepsizes play as important local acceleration components; see, e.g., \cite{gao2019parallelizable,gao2020orthogonalization,gao2021stop,xiao2020class,xiao2021penalty}. However, the application of the (A)BB stepsizes on the atomic structure relaxation problem remains to be investigated. In this work, we take the (A)BB stepsizes as the initial trial stepsizes to exploit their merits. 
	
	\par Without LM, the energy values produced by gradient descent methods with the (A)BB stepsizes do not necessarily decrease monotonically. This phenomenon reminds us of the possible divergence when starting the algorithms far from equilibria. Nevertheless, the monotone LM criterion, adopted by the conventional CG, can ruin the merits of the (A)BB stepsizes, since frequent calls of LM algorithms are entailed to assure the monotonic decrease of the energy and the local curvature of PES may be lost. Therefore, it is more sensible to design a \textit{nonmonotone} LM (NLM) criterion, accepting the original (A)BB stepsizes in most cases. 
	
	\par The existing NLM criteria primarily fall into two categories: the ``max'' type (MNLM) \cite{grippo1986nonmonotone} and the ``average'' type (ANLM) \cite{zhang2004nonmonotone}. In the $(k+1)$-th iteration, the MNLM criterion asks the LM algorithm to return a factor $r_k>0$ such that
	$$E\lrbracket{R_k+r_k\alpha_k^{\text{trial}}F_k}\le\bar E_k-c\cdot r_k\alpha_k^{\text{trial}}\snorm{F_k}_\F^2,$$
	where $\bar E_k:=\max_{i=k-m(k)+1}^{k}\lrbrace{E_i}$, $m(k)\in[0,M]$, both $c$ and $M$ are given positive constants, i.e., there is a sufficient reduction with respect to (w.r.t.) the largest one of the past $m(k)$ energies. The ANLM criterion requires $r_k$ to fulfill
	\begin{align}
		E\lrbracket{R_k+r_k\alpha_k^{\text{trial}}F_k}&\le C_k-c\cdot r_k\alpha_k^{\text{trial}}\snorm{F_k}_\F^2,\label{eqn:sufficient reduction}\\
		\inner{F\lrbracket{R_k+r_k\alpha_k^{\text{trial}}F_k},F_k}&\le\sigma\snorm{F_k}_\F^2,\label{eqn:curvature}
	\end{align}
	where $\sigma\in(0,1)$, the monitoring sequence~$\lrbrace{C_k}$ is updated through
	\begin{equation}
		C_{k+1}:=\frac{\eta_kQ_kC_k+E_{k+1}}{\eta_kQ_k+1},~Q_{k+1}:=\eta_kQ_k+1
		\label{eqn:surrogate sequence}
	\end{equation}
	with $C_0:=E_0$, $Q_0:=1$, and $\eta_k\in[0,1]$. Roughly in the ANLM criterion, a sufficient reduction w.r.t. $C_k$ is obligatory. It is straightforward to deduce that $C_k$ is a weighted average of the past energies $\{E_j\}_{j=0}^{k-1}$ and the current $E_{k}$. The construction of the monitoring sequence is completely free of a priori physical knowledge. Clearly, entailing reduction w.r.t. $\bar E_k$ and $C_k$ respectively instead of $E_k$, both MNLM and ANLM criteria are less restrictive than the monotone criterion adopted by the conventional CG because $\bar E_k$, $C_k\ge E_k$. Moreover, as noted in \cite{zhang2004nonmonotone}, the MNLM criterion is sensitive to the choice of $M$ in some contexts, and is often outperformed by the ANLM criterion. 
	
	\begin{figure}[htbp]
		\centering
		\includegraphics[width=.9\columnwidth]{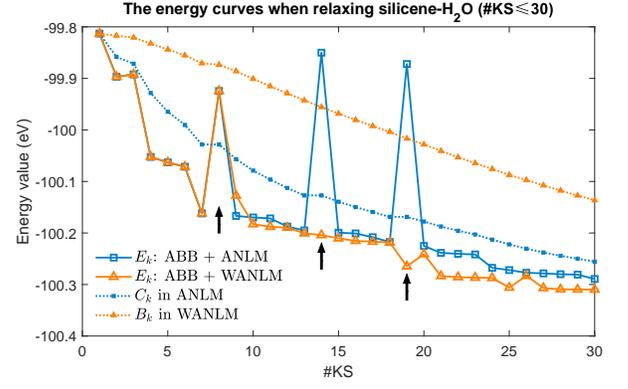}
		\caption{The energy curves when relaxing the silicene surface with one adsorbed H$_2$O molecule using gradient descent method equipped with the ABB stepsizes and ANLM ($\eta_k\equiv0.85$\footnote{This value is recommended by \cite{zhang2004nonmonotone}.}) or WANLM ($\mu_k\equiv0.05$) criterion. ``\#KS'' refers to the number of solved KS equations. We plot out the history with $\text{\#KS}\le30$. The blue solid line with square markers and orange line with triangle markers stand for energies produced by ``ABB + ANLM'' and ``ABB + WANLM'', respectively. The blue dashed line with square markers and orange line with triangle markers represent $\{C_k\}$ in ANLM and $\{B_k\}$ in WANLM, respectively. The black arrows indicate the \ab~calculations dedicated to the trials rejected by the ANLM criterion.}
		\label{fig:wanlm silicene}
	\end{figure}

    \par However, preliminary simulations show that the ANLM criterion is not loose enough for the (A)BB stepsizes during the atomic relaxation. We take the gradient descent direction armed with the ABB stepsizes and ANLM criterion to relax the silicene surface with one adsorbed H$_2$O molecule. From the short relaxation history in Fig. \ref{fig:wanlm silicene}, one can observe that the monitoring sequence $\{C_k\}$ (the blue dashed line with square markers) decreases fast, stays close to the energies $\{E_k\}$, and does not tolerate somewhat large increments. As a result, the combination ``ABB + ANLM'' solves 3 KS equations for the unaccepted trials during the short history, which are marked out by the black arrows in Fig. \ref{fig:wanlm silicene}. For example, the initial trial step at the 8-th iteration gives an energy value of -99.92 eV, which is higher than $C_7=$~-100.03 eV and thus is rejected by the ANLM criterion. 
    
    \par To this end, we devise a reweighted ANLM (WANLM) criterion, which differs from the ANLM criterion in the weights assigned to the past and current energies. Specifically, we denote the monitoring sequence in WANLM by $\lrbrace{B_k}$ and update it via
	\begin{equation}
		B_{k+1}:=\frac{B_k+\mu_kP_kE_{k+1}}{1+\mu_kP_k},~P_{k+1}:=1+\mu_kP_k
		\label{eqn:new surrogate sequence}
	\end{equation}
	with $B_0:=E_0$, $P_0:=1$, and $\mu_k\in[0,1]$. Then WANLM needs $r_k$ to satisfy both Eqs.~\eqref{eqn:sufficient reduction} and \eqref{eqn:curvature} with $C_k$, $Q_k$, and $\eta_k$ replaced by $B_k$, $P_k$, and $\mu_k$, respectively.
	
    \par The looseness of the WANLM criterion can be understood via the comparison with the ANLM criterion. Due to Eq.~\eqref{eqn:surrogate sequence} and $\eta_0\in[0,1]$, the coefficient of $E_1$ is $1/(1+\eta_0)\ge0.5$, indicating a large portion of $E_1$ in $C_1$ and a small gap $C_1-E_1$. Since the monitoring sequence $\{C_k\}$ decreases monotonically (see the appendix for a rigorous proof), it would then stay close to the energy values $\{E_k\}$, rendering the ANLM criterion stringent for the (A)BB stepsizes as reflected in Fig. \ref{fig:wanlm silicene}. In contrast, from Eq.~\eqref{eqn:new surrogate sequence} and $\mu_0\in[0,1]$, the coefficient of $E_1$ is $\mu_0/(1+\mu_0)\le0.5$. Therefore, the monitoring sequence $\{B_k\}$ in the WANLM criterion is relatively far away from the energies $\{E_k\}$ at the beginning and gradually approaches the latter one, leaving sufficient room to accept the (A)BB stepsizes. As shown in Fig. \ref{fig:wanlm silicene}, the first 7 energy values obtained by ``ABB + ANLM'' and ``ABB + WANLM'' are the same. But one can find a much larger distance from $\{E_k\}$ to $\{B_k\}$ than $\{C_k\}$, which enables the WANLM criterion to accept the 8-th initial trial step without invoking extra \ab~calculations and retain the merits of the ABB stepsizes. All the subsequent iterations of ``ABB + WANLM'' benefit from this characteristic and progress to lower energy values faster. Regarding safety, it can be verified that the sufficient reduction in terms of the atomic force norm holds for the monitoring sequence $\{B_k\}$. One can then establish the convergence to equilibria, regardless of initial configurations, after noticing the lower boundedness of $\{B_k\}$. For more details, please refer to the appendix. 
    
    \par The gradient descent method combining the (A)BB stepsizes with the WANLM criterion is referred to as WANBB. We set the initial trial stepsize to be $4.8\times10^{-2}$ \AA$^2$/eV at the first iteration, and take the alternating strategy described in Eq.~\eqref{eqn:abb} afterwards because it leads to better performance than BB1 and BB2 in our simulations. Before invoking an LM algorithm, we take absolute value in case that $\inner{S_k,Y_k}<0$ and make truncation for safeguard:
	\begin{equation}
		\alpha_k^{\text{ABB}}:=\min\lrbrace{\abs{\alpha_k^{\text{ABB}}},\max\lrbrace{-\log_{10}(\Vert F_k\Vert_{2,\infty}),1.0}}.
		\label{eqn:truncation}
	\end{equation}
	As is shown in the appendix, the convergence of WANBB to equilibria can be attributed to both Eqs.~\eqref{eqn:sufficient reduction} and \eqref{eqn:curvature}. For the sake of implementation, it is more efficient to fullfill only Eq.~\eqref{eqn:sufficient reduction}. In WANLM, we set $\mu_k\equiv0.05$, $c=10^{-4}$. The LM algorithm used by WANBB forms quadratic or cubic interpolations based on recent information of PES; for details see \cite[Chapter 3]{nocedal2006numerical}. 
	
    \par We have implemented WANBB in the in-house plane-wave code CESSP \cite{fang2016on,gao2017parallel,zhou2018applicability,fang2019implementation}. The exchange-correlation energy is described by the generalized gradient approximation \cite{perdew1996generalized}. Electron-ion interactions are treated with the projector augmented-wave (PAW) potentials based on the open-source ABINIT Jollet-Torrent-Holzwarth dataset library in the PAW-XML format \cite{jollet2014generation}. The total energies are calculated using the Monkhorst-Pack mesh \cite{monkhorst1976special} with the $k$-mesh spacing of 0.2 \AA$^{-1}$. The typical plane-wave energy cutoffs are around 500\textasciitilde600 eV. The KS equations corresponding to the updated atomic positions are solved by the preconditioned self-consistent field (SCF) iteration \cite{zhou2018applicability}.
    
    \par In addition to WANBB, we test the performances of three relaxation methods, including the conventional CG \cite{nocedal2006numerical,shewchuk1994introduction}, LBFGS \cite{liu1989limited}, and DIIS \cite{pulay1980diis}. By the way, DIIS is the so-called quasi-Newton method in some software such as VASP \cite{kresse1996efficiency, kresse1996efficient}, which can be derived from Broyden's method \cite{johnson1988modified}. The implementations of CG and DIIS follow those in \cite{nocedal2006numerical,liu1989limited}, and LBFGS is realized by the transition state tool VTST \cite{vtst}. CG employs Brent's method \cite{brent2013algorithms} as the LM algorithm. 
    
    \par The benchmark for performance test contains nearly 80 systems from various categories, including organic molecules, metallic systems, semiconductors, surface-molecule adsorption systems, and ABX$_3$ perovskites. The number of atoms ranges from 2 to 429. Some of them are available online, e.g., in Materials Project \cite{Jain2013} and Organic Materials Database \cite{borysov2017organic}. In the beginning, these systems are at their ideal crystal positions for defects, heterostructures, and substitutional alloys, or simply place a molecule on top of a surface. These are the typical initial configurations used by researchers in their simulations. More information on the benchmark can be found in Supplementary Material \footnote{See Supplemental Material at [URL will be inserted by publisher] for information on the benchmark set and detailed numerical results.}. The atomic structure relaxation is terminated if either the maximum atomic force falls below 0.01 eV/\AA~or the number of the solved KS equations (\#KS) arrives at 1000.  The convergence criterion for the SCF iteration is $10^{-5}$ eV.

	\par We evaluate the performances of the relaxation methods in two ways. One is the number of the solved KS equations (\#KS) needed to reach the 0.01 eV/\AA~force tolerance. The other is the running time (CPU) of the atomic structure relaxation. We adopt Dolan and Mor{\'e}'s performance profile \cite{dolan2002benchmarking} for an overall comparison among CG, WANBB, DIIS, and LBFGS. To this end, for any method indicated by $m\in\lrbrace{\text{CG,~WANBB,~DIIS,~LBFGS}}$, we denote its \#KS and CPU for relaxing the $n$-th system by $\text{\#KS}_{n,m},~\text{CPU}_{n,m}$, respectively \footnote{For visualization, we reset $\text{\#KS}_{n,m}=2000$ and $\text{CPU}_{n,m}=864000$ when the method $m$ fails on the $n$-th system.}, define ratios
	\begin{equation*}
		r_{n,m}^{\text{KS}}:=\frac{\text{\#KS}_{n,m}}{\min_m\text{\#KS}_{n,m}},\quad r_{n,m}^{\text{CPU}}:=\frac{\text{CPU}_{n,m}}{\min_m\text{CPU}_{n,m}},
	\end{equation*}
    and for any $\omega\ge0$, let
	\begin{align*}
		\pi_m^{\text{KS}}(\omega)&:=\frac{\text{the number of systems satisfying }r_{n,m}^{\text{KS}}\le\omega}{\text{the number of systems}},\\
		\pi_m^{\text{CPU}}(\omega)&:=\frac{\text{the number of systems satisfying }r_{n,m}^{\text{CPU}}\le\omega}{\text{the number of systems}}.
	\end{align*}
	The quantities $\pi_m^{\text{KS}}(\omega)$ and $\pi_m^{\text{CPU}}(\omega)$ estimate respectively the probabilities of $r_{n,m}^{\text{KS}}\le\omega$ and $r_{n,m}^{\text{CPU}}\le\omega$ over the benchmark \footnote{To give a fair and expressive numerical comparison, we (i) filter out the the systems where the differences between the maximum and the minimum of the converged energies per atom by the four solvers are larger than 1 meV; (ii) retain the systems at which certain solvers fail. The final number of systems used for Fig. \ref{fig:perf} is 68.}. For example, $\pi_{\text{WANBB}}^{\text{CPU}}(1)$ yields the portion of the systems where WANBB is the fastest. And $\pi_{\text{WANBB}}^{\text{CPU}}(2)$ means the estimated probability that, for one system, the running time of WANBB is not more than twice the shortest running time needed by CG, WANBB, DIIS, and LBFGS. In terms of the \#KS, $\pi_m^{\text{KS}}(\omega)$ can be interpreted in a similar way. We plot $\pi_m^{\text{KS}}(\omega)$ and $\pi_m^{\text{CPU}}(\omega)$ as a function of $\omega$ for each method in Fig. \ref{fig:perf}. Basically, the larger the area under the curve, the better the overall performance of the method. It is easy to see that the overall performance of WANBB ranks the best. More specifically, one can tell from the supplementary Tables I and II \cite{Note1} that the average CPU speedup factors of WANBB over CG, DIIS, and LBFGS are 1.51, 1.21 and 1.16, respectively.
	\begin{figure}[htbp]
		\centering
		\includegraphics[width=\columnwidth]{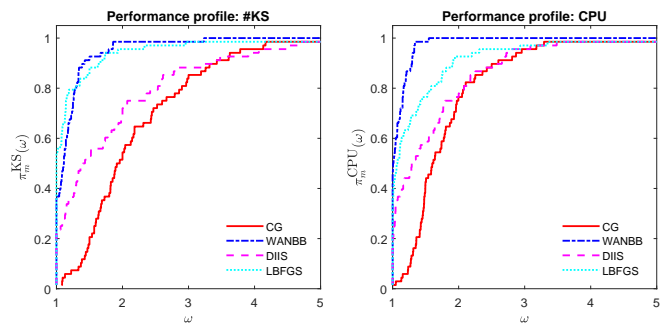}
		\caption{The \#KS (left) and CPU (right) performance profiles of CG, WANBB, DIIS, and LBFGS. }
		\label{fig:perf}
	\end{figure}

    \par As shown in Fig. \ref{fig:perf}, though sometimes requiring solving more KS equations than LBFGS, WANBB turns out to be faster than LBFGS in most tested systems. LBFGS may make larger atomic displacements due to constant stepsize, which results in underestimating the quality of initial wave functions and slowing down the SCF iteration. Later, we will explain it in more detail using the relaxation of the PuO$_2$ surface with one adsorbed H$_2$O molecule (see Fig. \ref{fig:scf}).

    \par In addition to the efficiency, the robustness is also worth the whistle. As shown in the supplementary Tables I and II \cite{Note1}, WANBB converges in all the benchmark systems, while CG fails in 1 case due to the breakdown of Brent's method, DIIS and LBFGS fail in 8 cases and 1 case, respectively, due to divergence. This kind of robustness is theoretically guaranteed by the convergence property of WANBB proved in the appendix.
    
    \par Taking the performance of CG as the baseline, we illustrate the acceleration obtained by WANBB for each system, as depicted in Fig. \ref{fig:speedup}. We have an average CPU speedup factor around 1.51, and a factor above 2.0 on about 20\% of systems. According to the supplementary Tables I and II \cite{Note1}, CG takes on average 117.68 calls of solving the KS equations to reach the force tolerance. Thus a speedup factor of 1.5\textasciitilde2 can mean a significant saving of computational overhead. This saving is mainly attributed to the enhanced \ab~LM process of WANBB since the distribution of \#KS speedup is consistent with that of CPU speedup in Fig. \ref{fig:speedup}. Moreover, we also count the percentage of the LM steps invoked by WANBB and CG for unaccepted trials and draw the frequency distribution over the benchmark in Fig. \ref{fig:frequency_ls}. The average LM step percentages for unaccepted trials are 59.09\% and 1.47\% for CG and WANBB, respectively. It is obvious from these statistics that WANBB escapes from the LM process quickly, usually without invoking LM algorithms at all.
	
	\begin{figure}
		\centering
		\includegraphics[width=\columnwidth]{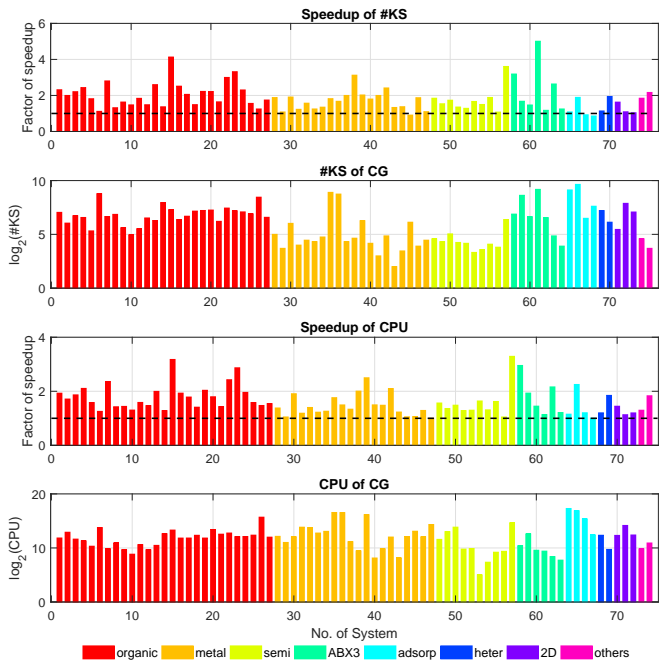}
		\caption{The speedup factors of \#KS and the \#KS of CG (the upper two), the speedup factors of CPU and the CPU of CG (the lower two). The 1.0 baselines are marked out by bold black dashed lines. The system where CG fails is not included.}
		\label{fig:speedup}
	\end{figure}

	\begin{figure}[htbp]
		\centering
		\includegraphics[width=\columnwidth]{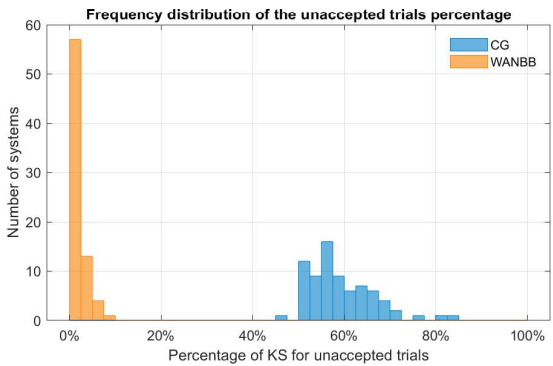}
		\caption{The frequency distributions of the LM step percentages of WANBB and CG for unaccepted trials.}
		\label{fig:frequency_ls}
	\end{figure}

	\par The following discussion is concentrated on the performance comparison between CG and WANBB for each category. In Fig. \ref{fig:speedup class size}, we plot the mean speedup factors by category and the number of atoms, respectively. All description below is supported by Tables I and II in Supplementary Material \cite{Note1}. Whenever speaking of ``speedup'' below, we are talking about the CPU speedup if not specified. 
	
	\begin{figure}[htbp]
		\centering
		\includegraphics[width=\columnwidth]{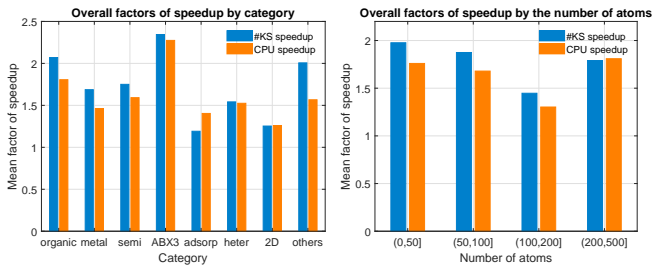}
		\caption{Mean factors of speedup by category (left) and the number of atoms (right).}
		\label{fig:speedup class size}
	\end{figure}

	\vskip 0.1cm
	
    \par \textit{Organic molecules}. This class of systems includes various alkanes and amino acids. The amino acids contain common organic atomic types, e.g., oxygen, nitrogen, sulfur, etc. The number of atoms in this class ranges from 26 to 106. Most of the initial configurations are available in Organic Materials Database \cite{borysov2017organic}. Compared with CG, WANBB achieves a mean speedup factor of 1.81 for this class and speedup factors greater than 2 on over 20\% of systems. In the case of Glutamate, the acceleration factor even reaches 3.18, which means a saving of nearly 70\% computational overhead.
	
	\vskip 0.1cm
	
	\par \textit{Metallic systems}. We consider Ag$_{16}$ cluster, Cu(111) surface with adsorbed organic molecules, and some alloys including Ni$_3$Al based intermetallic compound as well as high-entropy alloys (HEAs) such as FeCrNiCoAl. The speedup factors on these systems w.r.t. CG are mostly around 1.5\textasciitilde3. In particular, for Cu(111) surface with adsorbed organic molecules, CG needs to solve over 400 KS equations for atomic relaxations, while WANBB cuts the required \#KS roughly in half. It is also worth mentioning that WANBB can efficiently treat HEAs with potential strong local lattice distortion \cite{wu2020structural}. Actually, local lattice distortion is an essential factor in studying HEAs \cite{song2017local}, and can be described by the supercell method with \ab~atomic structure relaxation. In Fig. \ref{fig:rdf}, the local lattice distortion obtained by relaxing the body-centered cubic supercell of FeCrNiCoAl is illustrated by the smearing radial distribution function (RDF), in sharp contrast to the peaks at coordination shells in the case of the ideal lattice. On average in our benchmark, WANBB achieves a speedup of 1.54 over CG for HEAs.
	\begin{figure}[htbp]
		\centering
		\includegraphics[width=.7\linewidth]{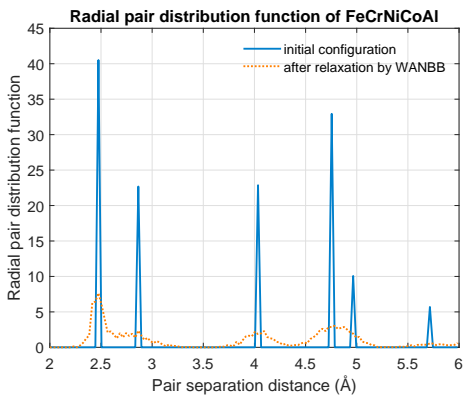}
		\caption{RDF of FeCrNiCoAl with 6 \AA~cutoff.}
		\label{fig:rdf}
	\end{figure}
	
	\vskip 0.1cm
	
	\par \textit{Semiconductors \& heterojunction systems}. We test semiconductors (CdSe, GaAs, Si, etc.) and heterojunction systems (GaAs-InAs). Moreover, we consider semiconductors with defects by constructing the bulk GaAs and Si supercells with some vacancies. In most of these systems, the acceleration factors are above 1.5, and sometimes exceed 3. We also investigate the performances of WANBB compared with CG on Si supercells with increasing system sizes; see Fig. \ref{fig:scalability}. As the system size increases, WANBB is more stable in terms of \#KS, and its CPU advantage becomes more prominent.
	\begin{figure}[htbp]
		\centering
		\includegraphics[width=\columnwidth]{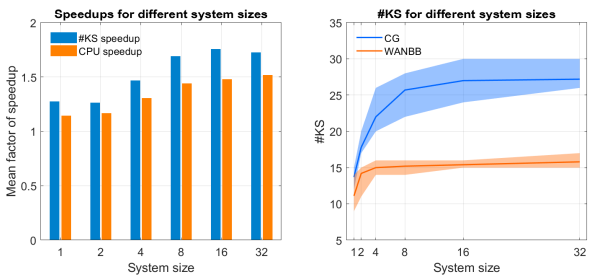}
		\caption{Left: the average \#KS and CPU speedup factors of WANBB over CG on $N\times1\times1$ Si supercells containing from 8 to 256 atoms, where each size has 10 random samples and the shift on each atom away from equilibrium is no larger than 0.1 \AA. Right: the \#KS of WANBB and CG on $N\times1\times1$ Si supercells, where the solid lines represent the average values, the upper and lower boundaries of the shaded area stand for the maximum and minimum \#KS, respectively.}
		\label{fig:scalability}
	\end{figure}
	
	\vskip 0.1cm
	
	\par \textit{ABX$_3$ perovskite systems}. This type of systems is fairly popular in recent solar cell studies. We consider the cases where ``A'' = CH$_3$NH$_3$, CH(NH$_2$)$_2$, or Cs, ``B'' = Pb or Sn, ``X'' = Br, Cl, or I. Their initial configurations come from Materials Project \cite{Jain2013}. Due to the existence of the organic molecules and their relative weak interaction with the inorganic framework, the atomic relaxations of such systems are quite challenging. Taking CH$_3$NH$_3$SnI$_3$ for instance, CG solves 581 KS equations, while WANBB entails only about one fifth of the expenditure. 
	
	\vskip 0.1cm
	
	\par \textit{Surface-molecule adsorption systems}. Several surface-molecule adsorption systems have been tested, including a 2D silicene surface with one adsorbed H$_2$O molecule (total 21 atoms), a 2D MoS$_2$ surface with one adsorbed NH$_3$ molecule (total 16 atoms), 2D PbS(001) and (111) surfaces with one adsorbed oleic acid molecule ((001) total 254 atoms, (111) total 429 atoms), and a 2D PuO$_2$ surface with one adsorbed H$_2$O molecule (total 195 atoms). The results show that CG often undergoes quite a long way until convergence in this category, and even fails in one case due to the breakdown of Brent's method. On the contrary, WANBB converges in all cases and usually costs less time; it is even twice as fast as CG on PbS(001) surface with one adsorbed oleic acid molecule. One interesting observation is that WANBB occasionally takes equal or a few more \#KS than CG, even though it is less time-costly. This is also observed when comparing WANBB and LBFGS. To explain, we plot the cumulative SCF iteration numbers of CG, WANBB, and LBFGS when relaxing the PuO$_2$ surface with one adsorbed H$_2$O molecule; see Fig. \ref{fig:scf}. 
	\begin{figure}[htbp]
		\centering
		\includegraphics[width=.8\columnwidth]{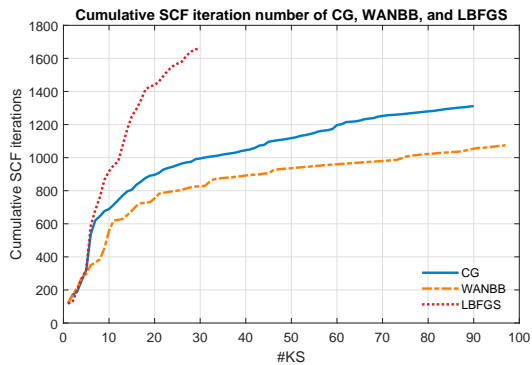}
		\caption{Cumulative SCF iteration numbers of CG, WANBB, and LBFGS on PuO$_2$ surface with one adsorbed H$_2$O molecule.}
		\label{fig:scf}
	\end{figure}
	One could observe that CG and LBFGS often consume more SCF iterations per KS equation than WANBB. In fact, CG and LBFGS make larger atomic displacements, resulting in degrading initial wave functions and slowing down the SCF iteration. On the other hand, Fig. \ref{fig:scf} also reveals that CG and LBFGS produce more efficient search directions, yielding fewer iteration numbers. This weakness, however, does not prevent WANBB from being a more favorable alternative in the context of \ab~simulations, because the cut down on the LM process obviously brings much more benefits. 
	
	\vskip 0.1cm
	
	\par \textit{2D and other systems}. The 2D materials in our benchmark include C$_{96}$, Si$_{32}$, and Germanene on hexagonal boron-nitride. The average speedup factors of WANBB over CG on these systems are 1.26. We also include systems of other type, e.g., cerium-iron oxide and potassium chloride. The former one is of practical usage for NO$_{\text{x}}$ reduction. WANBB reduces the numbers of solved KS equations by half on both of them.
	
	\vskip 0.1cm

	\par We recognize that the inefficiency of the LM process is an important obstacle to the high-performance execution of \ab~atomic structure relaxation. With the most widely used CG, nearly 60\% of the computational overhead on average is attributed to the unaccepted trials in the LM process. To this end, we propose a force-based gradient descent method called WANBB, where the initial trial (A)BB stepsizes grasp well the local curvature of PES and the unrestrictive WANLM criterion often accepts the initial trials without calling extra \ab~calculations. The robustness of WANBB, i.e., the convergence to equilibria regardless of initial configurations, is theoretically guaranteed. The numerical simulations on the benchmark containing nearly 80 various systems reveal the average 1.47\% computational cost for unaccepted trials in WANBB and demonstrate its universal speedups as well as superior robustness over CG, DIIS, and LBFGS. 
	
	\par Our work concentrates on the LM process. Nevertheless, efficient search directions together with subtly designed LM can be another great boon. This topic merits futural study. 
	
	\vskip 0.1cm 
	
	\par \textbf{Acknowledgements}
	\par The authors acknowledge Jun Fang for the help with CESSP, and thank Zhen Yang and Lifang Wang for the help on the construction of the benchmark. The computations were (partly) done on the high performance computers of State Key Laboratory of Scientific and Engineering Computing, Chinese Academy of Sciences. This work was partially supported by the National Natural Science Foundation of China under Grant Nos. 12125108, 11971466, 11991021, 11991020, 12021001, and 12288201, Key Research Program of Frontier Sciences, Chinese Academy of Sciences (No. ZDBS-LY-7022), the CAS AMSS-PolyU Joint Laboratory in Applied Mathematics, and the Science Challenge Project under Grant No. TZ2018002.
	
	\appendix*
	
	\section{Convergence Analysis of WANBB}\label{appsec:convergence}
	
	\par In this part, we analyze the convergence property of WANBB. Before moving on, we define the level set
	$$\mathcal{L}:=\lrbrace{R:E(R)\le E_0}.$$ 
	For the ease of reference, we restate the nonmonotone conditions in WANLM as follows: in the $(k+1)$-th iteration, $R_{k+1}$ satisfies
	\begin{align}
		E_{k+1}&\le B_k-c\cdot\alpha_k\snorm{F_k}_\F^2,\label{eqn:sufficient decrease new}\\
		\inner{F_{k+1},F_k}&\le\sigma\snorm{F_k}_\F^2.\label{eqn:curvature new}
	\end{align}
	By the arguments in \cite[Lemma 1.1]{zhang2004nonmonotone}, we know that the stepsize $\alpha_k$ satisfying both Eqs.~\eqref{eqn:sufficient decrease new} and \eqref{eqn:curvature new} always exists and that $B_k$ is a convex combination of $\lrbrace{E_i}_{i=0}^k$. 
	
	\par We first show that the atomic position sequence generated by WANBB remains in $\mathcal{L}$. 
	
	\begin{lem}\label{lem:level set}
		Let $\lrbrace{R_k}$ be the atomic position sequence generated by \textup{WANBB}~with $\mu_k\in[\mu_{\min},\mu_{\max}]\subseteq[0,1]$. Then we have $\lrbrace{R_k}\subseteq\mathcal{L}$. 
	\end{lem}
	
	\begin{proof}
		We prove by contradiction. Suppose there exists $l_1>0$ such that $E_{l_1}>E_0$. By the condition \eqref{eqn:sufficient decrease new}, 
		$$E_{l_1}\le B_{l_1-1}-c\cdot\alpha_{l_1-1}\snorm{F_{l_1-1}}_\F^2\le B_{l_1-1}.$$
		Since $B_{l_1-1}$ is a convex combination of $\lrbrace{E_0,\ldots,E_{l_1-1}}$, there must exist some $l_2\le l_1-1$ for which $E_{l_1}\le E_{l_2}$ holds. We can repeat the above arguments to obtain a nonnegative sequence $\lrbrace{l_i}$ which is finite for its strict monotonicity and ends up with 0. Therefore, we get $E_0<E_{l_1}\le E_{l_2}\le\cdots\le E_0$, a contradiction.
	\end{proof}
	
	\par From the update formula \eqref{eqn:new surrogate sequence}, one could obtain that the surrogate sequence $\{B_k\}$ is always lower bounded by the potential energy sequence $\{E_k\}$. 
	
	\begin{lem}\label{lem:E le B}
		Let $\{R_k\}$ be the atomic position sequence generated by WANBB with $\mu_k\in[\mu_{\min},\mu_{\max}]\subseteq[0,1]$. Then we have $E_k\le B_k$ for any $k\ge0$.
	\end{lem}
	
	\begin{proof}
		The claim is true for $k=0$ by definition. For $k\ge1$, by the update formula \eqref{eqn:new surrogate sequence} of $B_{k+1}$, one has
		$$B_{k}=\frac{B_{k-1}+\mu_{k-1}Q_{k-1}E_{k}}{1+\mu_{k-1}Q_{k-1}}\ge\frac{E_k+\mu_{k-1}Q_{k-1}E_k}{1+\mu_{k-1}Q_{k-1}}=E_k,$$
		where the inequality follows from the condition \eqref{eqn:sufficient decrease new}. The proof is complete.
	\end{proof}
	
	\par With Lemmas \ref{lem:level set} and \ref{lem:E le B} in place, we are ready to show the convergence of WANBB. The proof relies on the conditions \eqref{eqn:sufficient decrease new} and \eqref{eqn:curvature new} as well as the update formula \eqref{eqn:new surrogate sequence}. Specifically, the condition \eqref{eqn:curvature new} ensures a uniform positive lower bound of stepsizes, while Eqs.~\eqref{eqn:sufficient decrease new} and \eqref{eqn:new surrogate sequence} guide us to a sufficient reduction over $\{B_k\}$.
	
	\begin{thm}\label{thm:WANBB global convergence}
		Suppose the potential energy is bounded from below and $F$ is Lipschitz continuous, with modulus $L>0$, on $\mathcal{L}$, namely,
		$$\snorm{F(R)-F(\bar R)}_\F\le L\snorm{R-\bar R}_\F,~\forall~R,\bar R\in\mathcal{L}.$$
		Let $\lrbrace{R_k}$ be the atomic position sequence generated by \textup{WANBB}~with $\mu_k\in[\mu_{\min},\mu_{\max}]\subseteq(0,1]$. Then we have 
		$$\lim_{k\to\infty}\snorm{F_k}_\F=0.$$
		That is to say, the atoms~will approach equilibria, regardless of initial configurations. 
	\end{thm}
	
	\begin{proof}
		Substracting $\snorm{F_k}_\F^2$ from both side of Eq.~\eqref{eqn:curvature new}, we have
		$$\inner{F_k-F_{k+1},F_k}\ge\lrbracket{1-\sigma}\snorm{F_k}_\F^2.$$
		Combining Lemma \ref{lem:level set}, the Lipschitz continuity of $F$ in $\mathcal{L}$, and the fact $\sigma\in(0,1)$, we can further derive 
		\begin{equation*}
			\begin{aligned}
				(1-\sigma)\snorm{F_k}_\F^2&\le\snorm{F_k-F_{k+1}}_\F\cdot\snorm{F_k}_\F\\
				&\le L\snorm{R_k-R_{k+1}}\cdot\snorm{F_k}\\
				&=\alpha_kL\snorm{F_k}_\F^2,
			\end{aligned}
		\end{equation*}
		and hence a lower bound for stepsize whenever $F_k$ does not vanish:
		\begin{equation}
			\alpha_k\ge\frac{1-\sigma}{L}.
			\label{eqn:lower bd for step}
		\end{equation}
		Plugging Eq.~\eqref{eqn:lower bd for step} into condition \eqref{eqn:sufficient decrease new}, we obtain
		\begin{equation}
			E_{k+1}\le B_k-c\cdot\alpha_k\snorm{F_k}_\F^2\le B_k-\frac{c(1-\sigma)}{L}\snorm{F_k}_\F^2.
			\label{eqn:recursive of energy}
		\end{equation}
		Let $\beta:=c(1-\sigma)/L$. By the update formula \eqref{eqn:new surrogate sequence} of $B_{k+1}$, we obtain a recursion for the surrogate sequence using Eq.~\eqref{eqn:recursive of energy}:
		\begin{equation}
			\begin{aligned}
				B_{k+1}&=\frac{B_k+\mu_kQ_kE_{k+1}}{1+\mu_kQ_k}\\
				&\le\frac{B_k+\mu_kQ_k\lrbracket{B_k-\beta\snorm{F_k}_\F^2}}{Q_{k+1}}\\
				&=B_k-\beta\frac{\mu_k Q_k}{Q_{k+1}}\snorm{F_k}_\F^2.
			\end{aligned}
			\label{eqn:recursive of surrogate}
		\end{equation}
		A byproduct of Eq.~\eqref{eqn:recursive of surrogate} is the monotonicity of $\lrbrace{B_k}$. Since $E$ is bounded from below on $\mathcal{L}$ and $B_k\ge E_k$ for any $k\ge0$ (by Lemma \ref{lem:E le B}), $\lrbrace{B_k}$ is lower bounded as well. Consequently, by Eq.~\eqref{eqn:recursive of surrogate}, 
		\begin{equation}
			\sum_{k=0}^\infty\frac{\mu_kQ_k\snorm{F_k}_\F^2}{Q_{k+1}}\le\frac{1}{\beta}\sum_{k=0}^\infty\lrbracket{B_k-B_{k+1}}<+\infty.
			\label{eqn:summation}
		\end{equation}
		Following from the update formula of $\lrbrace{Q_k}$ and the fact $\mu_{\max}\ge\mu_k\ge\mu_{\min}>0$, $Q_k\ge1$, we can derive
		$$\frac{\mu_kQ_k}{Q_{k+1}}\ge\frac{\mu_{\min}Q_k}{\mu_kQ_k+1}\ge\frac{\mu_{\min}}{\mu_k+1/Q_k}\ge\frac{\mu_{\min}}{\mu_{\max}+1},$$
		which, together with Eq.~\eqref{eqn:summation}, yields the summability of $\snorm{F_k}$. Hence, $\snorm{F_k}\to0$ as $k\to\infty$ as desired. 
	\end{proof}
	
	\normalem
	%\bibliographystyle{apsrev4-2}
	%\bibliography{reference}

\end{document}